\documentclass[conference]{IEEEtran}
\IEEEoverridecommandlockouts
\usepackage{amsmath,amssymb,amsfonts}
\usepackage{algorithmic}
\usepackage{graphicx}
\usepackage{textcomp}
\usepackage{array}
\usepackage{diagbox}
\usepackage{bm}
\usepackage{stfloats}
\usepackage{xcolor}
\usepackage{romannum}
\usepackage{comment}
\usepackage{siunitx}
\definecolor{purple}{RGB}{128,0,128}

\newcommand{\new}{\textcolor{blue}}
\newcommand{\rev}{\textcolor{black}}

\usepackage[style=ieee, doi=false, isbn=false, url=false]{biblatex}
\newtheorem{proposition}{{Proposition}}

\newtheorem{remark}{{Remark}}
\newtheorem{corollary}{{Corollary}}
\newenvironment{proof}{{\noindent\it Proof:}}{\hfill $\square$\par}

\begin{document}

\title{Optimized Parameter Design for Channel State Information-Free Location Spoofing\\
} 
\author{\IEEEauthorblockN{Jianxiu Li and
Urbashi Mitra}
{Department of Electrical and Computer Engineering, University of Southern California, CA, USA}\\
E-mail: \{jianxiul, ubli\}@usc.edu\thanks{This work has been funded by one or more of the following: DOE DE-SC0021417, Swedish Research Council 2018-04359, NSF CCF-2008927, NSF RINGS-2148313, NSF CCF-2200221, NSF CIF-2311653, ARO W911NF1910269, ONR 503400-78050, and ONR N00014-15-1-2550.}}
\maketitle

\begin{abstract}
In this paper, an augmented analysis of a delay-angle information spoofing (DAIS) is provided for location-privacy preservation, where the location-relevant delays and angles are artificially shifted to obfuscate the eavesdropper with an incorrect physical location. A simplified misspecified Cram\'{e}r-Rao bound (MCRB) is derived, which clearly manifests that not only estimation error,  but also the geometric mismatch introduced by DAIS can lead to a significant increase in localization error for an eavesdropper. Given an assumption of the orthogonality among wireless paths, the simplified MCRB can be further expressed as a function of delay-angle shifts in a closed-form, which enables the more straightforward optimization of these design parameters for location-privacy enhancement. Numerical results are provided, validating the theoretical analysis and showing that the root-mean-square error for eavesdropper’s localization can be more than $150$ m with the optimized delay-angle shifts for DAIS.
\end{abstract}


\section{Introduction} 
Due to the large-scale proliferation of location-based services, accurate location estimation for user equipment (UE) has been widely investigated. With the proposed widespread deployment of millimeter wave (mmWave) multi-antenna systems, UE locations can be precisely estimated via advanced estimation algorithms
\cite{Shahmansoori,Zhou,FascistaMISOML,FascistaMISOMLCRB,Li}. However, these methods cannot distinguish between an \textit{authorized device} (AD) performing the localization or that of an \textit{unauthorized device} (UD).  Thus the nature of the wireless channel means that signals can be eavesdropped and UD's can localize \cite{Shahmansoori,Zhou,FascistaMISOML,FascistaMISOMLCRB,Li} thus leaking location-relevant information \cite{Ayyalasomayajula,li2023fpi}.

To preserve privacy at the physical layer, the wireless channel itself or its statistics is usually exploited \cite{Tomasin2,Checa,,Goel,li2023fpi,yildirim2024deceptive,Ayyalasomayajula,DAIS,Ardagnaobfuscation,Goztepesurveyprivacy,dacosta2023securecomm}. Given accurate channel state information (CSI), artificial noise injection \cite{Tomasin2} and transmit beamforming \cite{Checa,Ayyalasomayajula} have been designed specifically for location-privacy enhancement, which either decreases the received signal-to-noise ratio (SNR) for UDs or hides the location-relevant delay-angle information. However, accurately acquiring CSI is expensive for resource-limited devices. To limit the location-privacy leakage without CSI, a fake path injection (FPI) is initially proposed in \cite{li2023fpi}. By virtually injecting a few fake paths with a precoder design, a statistically harder estimation problem is created for UDs \cite{Li,li2023fpi,Li2}. A similar deceptive jamming design is examined in \cite{yildirim2024deceptive}. 

Herein,  we consider an approach that seeks to de-stabilize the eavesdropper's (UD's) estimator while maintaining good performance for the AD.  To this end, there is some similarity between analyses of stability of the super-resolution problem \cite{Ankur,LiTIT,candes2014towards,LIACHA} and our examination of the Fisher Information Matrix (FIM) associated with the UD's localization task. In particular, we derive appropriate lower bounds that enable the assessment of the performance degradation of the UD. The prior stability work studies the impact of the implicit resolution limit between sources and the singularity of the FIM.
 
In particular, we build upon our prior work which designs
a delay-angle information spoofing (DAIS) strategy  in \cite{DAIS}, where the UE (Alice) can be \textit{virtually} moved to an incorrect location via shifting all the location-relevant delays and angles. We underscore that the methods in \cite{li2023fpi} and \cite{DAIS} are {\bf not} noise injection schemes which require the exchange of realizations of noise; in contrast, a handful of numbers (versus a waveform) are shared secretly between the UE and the AD.

While a {misspecified} Cram\'{e}r-Rao bound (MCRB) \rev{\cite{VuongMCRB,RichmondMCRB,Fortunatimismatchsurvey}} was developed in \cite{DAIS}, a strategy by which to design artificial delay-angle shifts to maximize location privacy was not developed as minimizing certain values is insufficient to guarantee an increase in privacy.  We seek to close this gap herein.

While the optimization of these shifts rely on the CSI, understanding such an association is a key step for a more robust, practical design in the future. We note that the analysis of optimizing MCRB is to further degrade UD's localization accuracy, in contrast to the analysis of the stability of the FIM \cite{MaximeISIT} or the secrecy rate for secure communications \cite{Oggier}. The main contributions of this paper are: 
\begin{enumerate} 
\item The MCRB on UD's localization derived in \cite{DAIS} is further simplified, providing a more clear insight on the obfuscation caused by DAIS \cite{DAIS}. 
\item Under an assumption of the orthogonality among the paths, the simplified MCRB is explicitly expressed as a function of the shifts for the delays and angles in a closed-form, suggesting the design of these key parameters.
\item Numerical results show that, in terms of UD's localization, the root-mean-square error (RMSE)  with the DAIS method proposed in \cite{DAIS} can be more than $150$ m if the delay-angle shifts are properly adjusted, validating theoretical analysis and the efficacy of DAIS.
\end{enumerate}

{
We use the following notation. Scalars are denoted by lower-case letters $x$ and column vectors by bold letters $\bm{x}$. The $i$-th element of $\bm{x}$ is denoted by $\bm{x}[i]$. Matrices are denoted by bold capital letters $\bm X$ and $\boldsymbol{X}[i, j]$ is the ($i$, $j$)-th element of $\bm{X}$. $\boldsymbol{I}_{l}$ is reserved for a $l\times l$ identity matrix. The operators $\|\bm  x\|_{2}$, $|x|$, $\mathfrak{R}\{x\}$, $\mathfrak{I}\{x\}$, $\lfloor{x}\rfloor$, and $\operatorname{diag}(\mathcal{A})$ stands for the $\ell_2$ norm of $\bm x$, the magnitude of $x$, the real part of $x$, the imaginary part of $x$, the largest integer that is less than $x$, and a diagonal matrix whose diagonal elements are given by $\mathcal{A}$, respectively. $(x)_{(t_1,t_2]}$ with $t_1<t_2$ is defined as $(x)_{(t_1,t_2]}\triangleq x-\left\lfloor\frac{x-t_1}{t_2-t_1}\right\rfloor(t_2-t_1)$
and $\mathbb{E}\{\cdot\}$ is used for the expectation of a random variable. The operators \rev{$\operatorname{Tr}(\cdot)$}, $(\cdot)^\mathrm{T}$, $(\cdot)^{\mathrm{H}}$ and $(\cdot)^{-1}$, are defined as \rev{the trace}, the transpose, the conjugate transpose, and the inverse of a vector or matrix, respectively.}

\section{System Model}\label{sec:systemmodel}
\begin{figure}[t]
\centering
\includegraphics[scale=0.455]{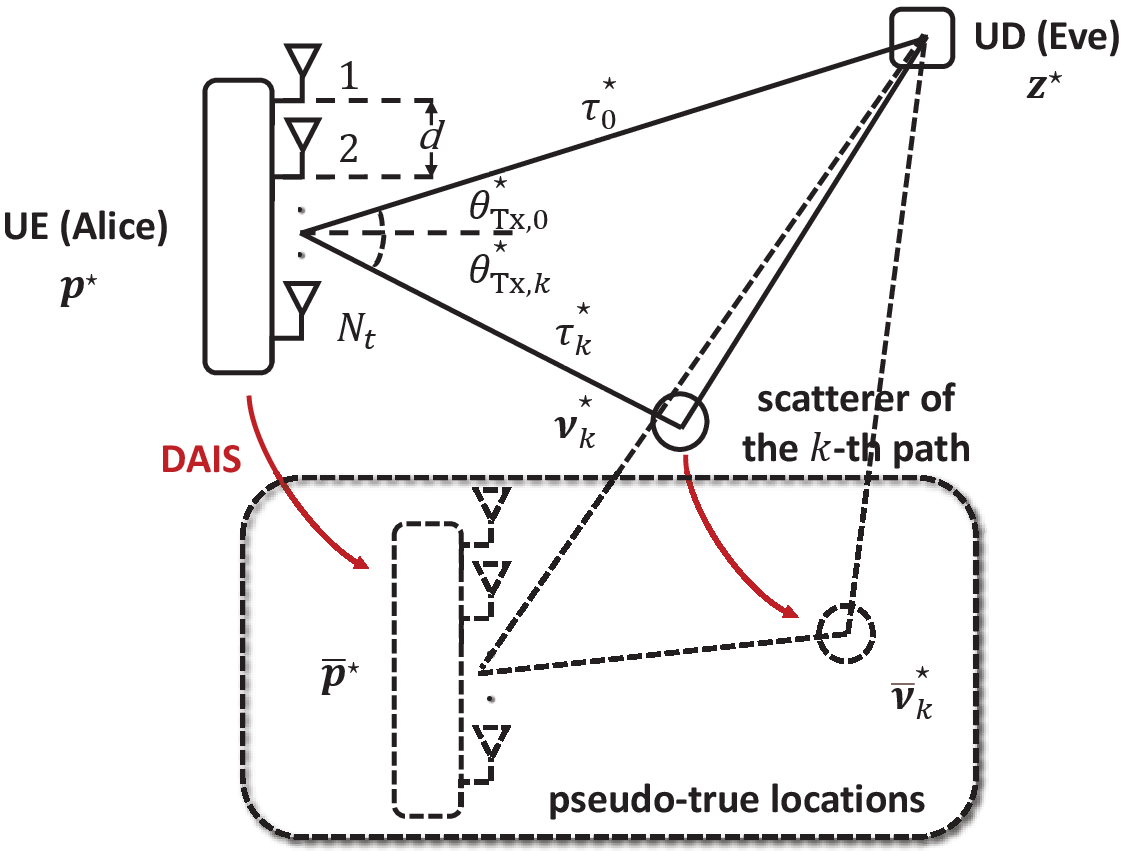}\vspace{-5pt}
\caption{{System model.}}\vspace{-10pt}
\label{fig:sys}
\end{figure}
We consider the system model of \cite{DAIS} as depicted in Figure \ref{fig:sys}. The UE (Alice) is at an unknown position $\bm p^{\star} = [p^{\star}_x, p^{\star}_y]^{\mathrm{T}} \in \mathbb{R}^{2\times 1}$ and is to be localized by an AD. Alice sends pilot signals over a public channel that can be overheard by the UD (Eve) \rev{at position $\bm z^{\star}=[z^{\star}_x, z^{\star}_y]^{\mathrm{T}}\in\mathbb{R}^{2\times 1}$ that is unknown to Alice}.
Assume Eve knows the pilot signals. Without any location-privacy preservation, Alice's location can be exposed if Eve leverages an effective localization algorithm, such as \cite{Shahmansoori,Zhou,FascistaMISOML,Li,FascistaMISOMLCRB}.

\subsection{Signal Model}
In an augmentation to \cite{DAIS}, we adopt  mmWave multiple-input-single-output (MISO) orthogonal frequency-division multiplexing (OFDM) signaling, where Alice is equipped with $N_t$ antennas, transmitting $G$ signals over $N$ sub-carrier, while Eve has only a single antenna. Denoting by $x^{(g,n)}$ and $\bm f^{(g,n)}\in\mathbb{C}^{N_t\times1}$ the $g$-th symbol transmitted over the $n$-th sub-carrier and the associated beamforming vector, respectively, we can express the pilot signal as $\boldsymbol{s}^{(g,n)}\triangleq \boldsymbol{ f}^{(g,n)}x^{(g,n)}\in\mathbb{C}^{N_t\times 1}
$. Assume that each pilot signal transmitted over the $n$-th sub-carrier are independent and identically distributed and we have $\mathbb{E}\{\boldsymbol{s}^{(g,n)}(\boldsymbol{s}^{(g,n)})^{\mathrm{H}}\}= \frac{1}{N_t}\bm I_{N_t}$. For the wireless channel between Alice and Eve, in addition to one available line-of-sight (LOS) path, it is assumed that there exist $K$ non-line-of-sight (NLOS) paths, produced by $K$ scatterers at an unknown position $\bm v^{\star}_k = [v^{\star}_{k,x},v^{\star}_{k,y}]^{\mathrm{T}}\in\mathbb{R}^{2\times 1}$, with $k = 1,2,\cdots,K$, respectively.  For notational convenience, let $k=0$ correspond to the LOS path and define $\bm v^{\star}_0\triangleq \bm z^{\star}$. Denote by $c$, $\lambda_c$, $d$, $B$, $T_s\triangleq\frac{1}{B}$, and $\bm\alpha(\theta)\triangleq\left[1, e^{-j\frac{2\pi d\sin(\theta)}{\lambda_c}}, \cdots, e^{-j\frac{2\pi(N_t-1)d\sin(\theta)}{\lambda_c}}\right]^{\mathrm{T}}\in\mathbb{C}^{N_t\times1}$, the speed of light, wavelength, distance between antennas designed as $d=\frac{\lambda_c}{2}$, bandwidth, sampling period, and a steer vector for an angle $\theta$, respectively. Given the narrowband assumption, \textit{i.e.}, $B\ll \varphi_c\triangleq\frac{c}{\lambda_c}$, the $n$-th sub-carrier public channel vector can be modeled as \cite{FascistaMISOMLCRB,DAIS}\vspace{-3pt}
\begin{equation}
\boldsymbol{h}^{(n)}\triangleq\sum_{k=0}^{K}\gamma^{\star}_k e^{\frac{-j 2\pi n\tau^{\star}_k}{N T_{s}}}\boldsymbol{ \alpha}\left(\theta^{\star}_{\mathrm{Tx},k}\right)^{\mathrm{H}}\in \mathbb{C}^{1\times N_t},\vspace{-3pt}
\label{eq:channel_subcarrier}
\end{equation} 
where $\gamma^{\star}_k$, $\tau^{\star}_k$, and $\theta^{\star}_{\mathrm{Tx},k}$ are the channel coefficient, the time-of-arrival (TOA), and the angle-of-departure (AOD) of the $k$-th path, respectively. Then, the received signal is given by 
\begin{equation}
{y}^{(g,n)}=\boldsymbol{h}^{(n)} \boldsymbol{s}^{(g,n)}+{w}^{(g,n)},\label{eq:rsignal}
\end{equation}
for $n = 0, 1, \cdots, N-1$ and $g = 1, 2, \cdots, G$, where ${w}^{(g,n)}\sim \mathcal{CN}({0},\sigma^2)$ is independent, zero-mean, complex Gaussian noise with variance $\sigma^2$. According to the geometry, the TOA and AOD of the $k$-th path can be derived as
\begin{equation}\label{eq:geometry}
\begin{aligned}
\tau^{\star}_{k} &=\frac{\left\|\boldsymbol{v}^{\star}_0-\boldsymbol{v}^{\star}_{k}\right\|_{2} +\left\|\boldsymbol{p}^{\star}-\boldsymbol{v}^{\star}_{k}\right\|_{2}} {c},\\
\theta^{\star}_{\mathrm{Tx}, k} &=\arctan \left(\frac{v^{\star}_{k,y}-p^{\star}_{y}} {v^{\star}_{k,x}-p^{\star}_{x}}\right)\rev{,}
\end{aligned}
\end{equation} 
where ${\tau^{\star}_k} \in(0,NT_s]$ and $\theta^{\star}_{\mathrm{Tx},k}\in(-\frac{\pi}{2},\frac{\pi}{2}]$ \cite{Li}. As in \cite{DAIS}, without of loss of generality, the orientation angle of the antenna array and the clock bias are assumed to be zero for the analysis in this paper. Herein, we provide further analysis over that in \cite{DAIS} to provide values for our spoofing design parameters.  Thus, we review the CSI-free DAIS method proposed in \cite{DAIS} before providing our extensions.

\subsection{Review of Delay-Angle Information Spoofing \cite{DAIS}}\label{sec:dis} 
{Denote by $\Delta_\tau$ and $\Delta_\theta$ two design parameters used for the DAIS design. To protect Alice's location from being accurately inferred by Eve with the location-relevant channel estimates, we shift the TOAs and AODs of the paths based on $\Delta_\tau$ and $\Delta_\theta$, respectively 
\begin{subequations}\label{eq:mismatchgeometry}
\begin{align} 
\bar{\tau}_{k}& = \left(\tau^{\star}_{k} + \Delta_\tau\right)_{(0,NT_s]}\\
\bar{\theta}_{\mathrm{Tx}, k}&= \arcsin\left({\left(\sin(\theta^{\star}_{\mathrm{Tx}, k})+\sin(\Delta_\theta)\right)_{\left(-1,1\right]}}\right).
\end{align}
\end{subequations} }

{To shift the associated TOAs and AODs, {a precoding matrix $\bm\Phi^{(n)}\in\mathbb{C}^{N_t\times N_t}$ is designed as
\begin{equation}\label{eq:daisbeamformer}
    {\bm \Phi}^{(n)} \triangleq  e^{-j\frac{2\pi n \Delta_\tau}{NT_s}} \operatorname{diag}\left(\bm \alpha\left(\Delta_\theta\right)^{\mathrm{H}}\right),
\end{equation}
for $n=0,1,\cdots,N-1$}. Accordingly, the signal received through the public channel can be re-expressed as 
\begin{equation}
\begin{aligned}
\bar{y}^{(g,n)}
&=\bar{\boldsymbol{h}}^{(n)}{\bm s}^{(g,n)}+{w}^{(g,n)},
\end{aligned}\label{daisrsignal}
\end{equation}
where $\bar{\boldsymbol{h}}^{(n)}\triangleq\sum_{k=0}^{K}\gamma^{\star}_k e^{\frac{-j 2\pi n\bar{\tau}_k}{N T_{s}}}\boldsymbol{ \alpha}\left(\bar{\theta}_{\mathrm{Tx},k}\right)^{\mathrm{H}}$ is a \textit{virtual channel} for the $n$-th sub-carrier, with distinct TOAs and AODs as compared with the original channel ${\boldsymbol{h}}^{(n)}$.}

{\section{Analysis of Eve’s Localization Accuracy with Delay-Angle Information Spoofing}}\label{sec:analysis}
{\subsection{Misspecified Cram\'{e}r-Rao Bound \cite{DAIS}}}
{Denote by $\bar{\bm\xi}\triangleq\left[\bar{\bm\tau}^\mathrm{T},\bar{\bm\theta}_{\mathrm{Tx}}^{\mathrm{T}},\mathfrak{R}\{\bm\gamma^{\star}\},\mathfrak{I}\{\bm\gamma^{\star}\}\right]^{\mathrm{T}}\in\mathbb{R}^{4(K+1)\times 1}$ a vector of the unknown channel parameters, where $\bar{\bm\tau}\triangleq\left[\bar{\tau}_0, \bar{\tau}_1\cdots, \bar{\tau}_K\right]^{\mathrm{T}}\in\mathbb{R}^{(K+1)\times 1}$, $\bar{\bm\theta}_{\mathrm{Tx}}\triangleq\left[\bar{\theta}_{\mathrm{Tx},0}, \bar{\theta}_{\mathrm{Tx},1}\cdots, \bar{\theta}_{\mathrm{Tx},K}\right]^{\mathrm{T}}\in\mathbb{R}^{(K+1)\times 1}$, and $\bm\gamma^{\star} \triangleq [\gamma^{\star}_0,\gamma^{\star}_1,\cdots,\gamma^{\star}_K]^{\mathrm{T}}\in\mathbb{R}^{(K+1)\times 1}$. Define $\bar{u}^{(g,n)} \triangleq \bar{\boldsymbol{h}}^{(n)}\bm s^{(g,n)}$. The FIM for the estimation of $\bar{\bm \xi}$ can be derived as \cite{Scharf} 
\begin{equation}\label{eq:FIMce}
\begin{aligned}
    &\bm J_{\bar{\bm \xi}}=  \\
    &\frac{2}{\sigma^2}\sum_{n=0}^{N-1}\sum_{g=1}^{G}\mathfrak{R}\left\{\left(\frac{\partial\bar{ u}^{(g,n)}}{\partial \bar{\bm\xi}}\right)^{\rev{\mathrm{H}}}\frac{\partial \bar{ u}^{(g,n)}}{\partial \bar{\bm\xi}}\right\}\in\mathbb{R}^{4(K+1)\times 4(K+1)}. 
\end{aligned}
\end{equation}
}

Let $\bar{\bm\eta}\triangleq[\bar{\bm\tau}^\mathrm{T},\bar{\bm\theta}_{\mathrm{Tx}}^{\mathrm{T}}]^{\mathrm{T}}\in\mathbb{R}^{2(K+1)\times1}
$ represent the location-relevant channel parameters and the FIM $\bm J_{\bar{\bm \xi}}$ is partitioned into $\bm J_{\bar{\bm \xi}}= \begin{bmatrix}\bm J_{\bar{\bm \xi}}^{(1)} &\bm J_{\bar{\bm \xi}}^{(2)}\\\bm J_{\bar{\bm \xi}}^{(3)}&\bm J_{\bar{\bm \xi}}^{(4)}\end{bmatrix}$, with $\bm J_{\bar{\bm \xi}}^{(m)}\in\mathbb{R}^{2(K+1)\times2(K+1)}$, for $m=1,2,3,4$. For the analysis of localization accuracy, we consider the channel coefficients as nuisance parameters and the effective FIM for the estimation of the location-relevant channel parameters $\bar{\bm\eta}$ is given by \cite{Tichavskyefim}
\begin{equation} \label{eq:efim}
    \bm J_{\bar{\bm \eta}} = \bm J_{\bar{\bm \xi}}^{(1)}-\bm J_{\bar{\bm \xi}}^{(2)}\left(\bm J_{\bar{\bm \xi}}^{(4)}\right)^{-1}\bm J_{\bar{\bm \xi}}^{(3)}\in\mathbb{R}^{2(K+1)\times 2(K+1)}. 
\end{equation}

{

Denote by \rev{$\bm\phi^{\star}\triangleq [(\boldsymbol{p}^{\star})^{\mathrm{T}},(\boldsymbol{v}^{\star}_1)^{\mathrm{T}},(\boldsymbol{v}^{\star}_2)^{\mathrm{T}},\cdots,(\boldsymbol{v}^{\star}_K)^{\mathrm{T}}]^{\mathrm{T}}\in\mathbb{R}^{2(K+1)\times1}$, $\hat{\bm\eta}_{\text{Eve}}$}, and $\hat{\bm \phi}_\text{Eve}$, \rev{a vector of the true locations of Alice and scatterers, Eve’s estimate of $\bar{\bm\eta}$}, and a \textit{misspecified-unbiased} estimator of $\bm\phi^{\star}$, respectively. {The misspecified Cram\'{e}r-Rao bound}  for the MSE of Eve’s localization is given by \cite{Fortunatimismatchsurvey,DAIS},
\begin{equation}\label{eq:EveLB}
\begin{aligned}
    &\mathbb{E}\left\{\left(\hat{\bm \phi}_\text{Eve}-{\bm \phi}^{\star}\right)\left(\hat{\bm \phi}_\text{Eve}-{\bm \phi}^{\star}\right)^{\mathrm{T}}\right\}\\
    &\quad\succeq \bm\Psi_{\bar{\bm \phi}^\star}\triangleq{{\bm A_{\bar{\bm \phi}^\star}^{-1}\bm B_{\bar{\bm \phi}^\star}\bm A_{\bar{\bm \phi}^\star}^{-1}}+{(\bar{\bm\phi}^\star-\bm\phi^{\star})(\bar{\bm\phi}^\star-\bm\phi^{\star})^\mathrm{T}}}, 
\end{aligned}
\end{equation}
where $\bm A_{\bar{\bm \phi}^\star}\in\mathbb{R}^{2(K+1)\times2(K+1)}$ and $\bm B_{\bar{\bm \phi}^\star}\in\mathbb{R}^{2(K+1)\times2(K+1)}$ are two generalized FIMs, defined as 
\begin{equation}\label{eq:A}
    \bm A_{\bar{\bm \phi}^\star}[r,l]\triangleq\mathbb{E}_{g_{\text{T}}(\hat{\bm\eta}_{\text{Eve}}|\bm\phi^\star)}\left\{\frac{\partial^2}{\partial\bar{\bm\phi}^\star[r]\partial\bar{\bm\phi^\star}[l]}\log g_{\text{M}}(\hat{\bm\eta}_{\text{Eve}}|\bar{\bm\phi}^\star)\right\},
\end{equation}
and\vspace{-3pt}
\begin{equation}\label{eq:B}
\begin{aligned}
&\bm B_{\bar{\bm \phi}^\star}[r,l]\\
&\triangleq\mathbb{E}_{g_{\text{T}}(\hat{\bm\eta}_{\text{Eve}}|\bm\phi^\star)}\left\{\frac{\partial\log g_{\text{M}}(\hat{\bm\eta}_{\text{Eve}}|\bar{\bm\phi}^\star) }{\partial\bar{\bm\phi}^\star[r]}\frac{\partial\log g_{\text{M}}(\hat{\bm\eta}_{\text{Eve}}|\bar{\bm\phi}^\star) }{\partial\bar{\bm\phi}^\star[l]}\right\},
\end{aligned} \vspace{-3pt}
\end{equation}
for $r,l = 1,2,\cdots,2(K+1)$, while \rev{ $g_{\text{T}}(\hat{\bm\eta}_{\text{Eve}}|\bm\phi^\star)$, $g_{\text{M}}(\hat{\bm\eta}_{\text{Eve}}|\bar{\bm\phi})$, and the vector $\bar{\bm\phi}^\star\triangleq [(\bar{\boldsymbol{p}}^\star)^{\mathrm{T}},(\bar{\boldsymbol{v}}^\star_1)^{\mathrm{T}},(\bar{\boldsymbol{v}}^\star_2)^{\mathrm{T}},\cdots,(\bar{\boldsymbol{v}}^\star_K)^{\mathrm{T}}]^{\mathrm{T}}\in\mathbb{R}^{2(K+1)\times1}$ represent the true and misspecified distributions of $\hat{\bm\eta}_{\text{Eve}}$, and the pseudo-true locations of Alice and scatterers, respectively.} Assume that the true and misspecified distributions of $\hat{\bm\eta}_{\text{Eve}}$ are characterized by the following two models, \vspace{-3pt}
\begin{subequations}\label{eq:truemismatchmodel}
    \begin{align}
    g_{\text{T}}(\hat{\bm\eta}_{\text{Eve}}|\bm\phi^\star):\quad\hat{\bm\eta}_{\text{Eve}}&=u(\bm\phi^\star)+\bm\epsilon,\\
    g_{\text{M}}(\hat{\bm\eta}_{\text{Eve}}|\bar{\bm\phi}):\quad\hat{\bm\eta}_{\text{Eve}}&=o(\bar{\bm\phi})+\bm\epsilon,\vspace{-3pt}
    \end{align} 
\end{subequations}
\noindent 
where $\bm\epsilon$ is a zero-mean Gaussian random vector with covariance matrix $\bm\Sigma_{\bar{\bm\eta}}\triangleq\bm J_{\bar{\bm\eta}}^{-1}$ and $u(\cdot)$ (or $o(\cdot)$) is a function mapping the location information $\bm\phi^\star$ (or $\bar{\bm\phi}$) to the location-relevant channel parameters $\bar{\bm\eta}$ according to the true geometric model defined in Equation \eqref{eq:mismatchgeometry} (or the misspecified geometric model defined in Equation \eqref{eq:geometry}). Due to the effect of potential \textit{phase wrapping} caused by the proposed DAIS scheme, the pseudo-true locations of Alice and scatterers are given by \cite{DAIS}, 
\vspace{-15pt}
\begin{subequations}\label{eq:ptrue}
    \begin{align} 
        \bar{\boldsymbol{p}}^{\star} &= \boldsymbol{z} - c\bar\tau_{k_{\min}}[\cos({\bar\theta}_{\mathrm{Tx},k_{\min}}),\sin(\bar\theta_{\mathrm{Tx},k_{\min}})]^{\mathrm{T}},\\
        \bar{v}^\star_{{k_{\min}},x} &=\frac{1}{2}\bar{b}^{\star}_{0}\cos(\bar{\theta}_{\mathrm{Tx},0})+\bar{p}^{\star}_{x},\\
        \bar{v}^\star_{{k_{\min}},y} &= \frac{1}{2}\bar{b}^{\star}_{0}\sin(\bar{\theta}_{\mathrm{Tx},0})+\bar{p}^{\star}_{y}, \\
        \bar{v}^\star_{k,x} &=\frac{1}{2}\bar{b}^{\star}_{k}\cos(\bar{\theta}_{\mathrm{Tx},k})+\bar{p}^{\star}_{x}, \quad \text{if} \ k\neq{k_{\min}},\\
        \bar{v}^\star_{k,y} &= \frac{1}{2}\bar{b}^{\star}_{k}\sin(\bar{\theta}_{\mathrm{Tx},k})+\bar{p}^{\star}_{y}, \quad \text{if} \ k\neq{k_{\min}},
    \end{align} \vspace{-3pt}
\end{subequations}
where $k_{\min}\triangleq\arg\min_k\bar{\tau}_k$
and 
\begin{equation}
    \begin{aligned}    
        \bar{b}^{\star}_{k}=\frac{\left(c\bar\tau_k\right)^2-\left(z_x-\bar{p}^{\star}_{x}\right)^2-\left(z_y-\bar{p}^{\star}_{y}\right)^2}{c\bar\tau_k-\left(z_x-\bar{p}^{\star}_{x}\right)\cos(\bar{\theta}_{\mathrm{Tx},k})-\left(z_y-\bar{p}^{\star}_{y}\right)\sin(\bar{\theta}_{\mathrm{Tx},k})}, 
    \end{aligned}
\end{equation}
with $k=1,2,\cdots,K$. {We emphasize 
that due to the phase wrapping, the estimate of the original LOS may not have the smallest delay;  however, when we perform localization, the path with smallest shifted TOA is assumed to be the LOS path. The notation $k_{\min} = k$ means the $k$-th path will be viewed as the LOS path.}
Though the degradation of Eve’s localization accuracy and the identifiability of the design parameters $\bm\Delta\triangleq[\Delta_\tau,\Delta_\theta]^\mathrm{T}\in\mathbb{R}^{2\times1}$ have been computed in \cite{DAIS}, it is not very clear how $\bm\Delta$ affects Eve’s performance. Therefore, the design of $\bm\Delta$ will be further investigated in this paper based on the analysis of MCRB derived in Equation \eqref{eq:EveLB}.
}

\subsection{Analysis of Misspecified Cram\'{e}r-Rao Bound}\label{sec:simplifiedMCRB}
{With key preliminaries in hand, we can now present the new contributions of the current work.}
{By substituting the pseudo-true locations of Alice and scatterers in Equation \eqref{eq:ptrue} into Equation \eqref{eq:EveLB}, we can simplify the MCRB for Eve's localization as follows.
\begin{corollary}\label{coro:EveLBsimplifed}
    Given the true and misspecified distributions of the estimated parameters $\hat{\bm \eta}_{\text{Eve}}$ in Equation \eqref{eq:truemismatchmodel}, the MSE of Eve’s localization can be bounded as 
\begin{equation}\label{eq:EveLBsimplified}
\begin{aligned}
    &\mathbb{E}\left\{\left(\hat{\bm \phi}_\text{Eve}-{\bm \phi}^{\star}\right)\left(\hat{\bm \phi}_\text{Eve}-{\bm \phi}^{\star}\right)^{\mathrm{T}}\right\}\\
    &\quad\quad\quad\quad\quad\succeq \bm\Psi_{\bar{\bm \phi}^\star}={{\bm J_{\bar{\bm \phi}^{\star}}^{-1}}+{(\bar{\bm\phi}^\star-\bm\phi^{\star})(\bar{\bm\phi}^\star-\bm\phi^{\star})^\mathrm{T}}}.
\end{aligned}
\end{equation}
\end{corollary}
\begin{proof}
    When the vector of the pseudo-true locations of Alice and scatterers $\bar{\bm\phi}^\star$ is set according to Equation \eqref{eq:ptrue}, it can be verified that $o(\bar{\bm\phi}^\star) = u(\bm\phi^{\star})$, which indicates that $g_{\text{T}}(\hat{\bm\eta}_{\text{Eve}}|\bm\phi^\star)=g_{\text{M}}(\hat{\bm\eta}_{\text{Eve}}|\bar{\bm\phi})$ for $\bar{\bm\phi}=\bar{\bm\phi}^\star$ according to Equation \eqref{eq:truemismatchmodel}. Hence, $\bm B_{\bar{\bm \phi}^\star}[r,l] = - \bm A_{\bar{\bm \phi}^\star}[r,l]$ holds due to 
\begin{equation}
\begin{aligned}
    &\bm A_{\bar{\bm \phi}^\star}[r,l]=\mathbb{E}_{g_{\text{T}}(\hat{\bm\eta}_{\text{Eve}}|\bm\phi^\star)}\left\{\frac{\partial^2}{\partial\bar{\bm\phi}^\star[r]\partial\bar{\bm\phi^\star}[l]}\log g_{\text{M}}(\hat{\bm\eta}_{\text{Eve}}|\bar{\bm\phi}^\star)\right\}\\
    &=\left(\frac{\partial}{\partial\bar{\bm\phi^\star}[l]}\left(\frac{\partial o(\bar{\bm\phi}^\star)}{\partial\bar{\bm\phi}^\star[r]}\right)\right)^{\mathrm{T}}\bm\Sigma_{\bar{\bm\eta}}^{-1}\mathbb{E}_{g_{\text{T}}(\hat{\bm\eta}_{\text{Eve}}|\bm\phi^\star)}\left\{\hat{\bm\eta}_{\text{Eve}}-o(\bar{\bm\phi}^\star)\right\}\\
    &\quad - \left(\frac{\partial o(\bar{\bm\phi}^\star)}{\partial\bar{\bm\phi}^\star[r]}\right)^{\mathrm{T}}\bm\Sigma_{\bar{\bm\eta}}^{-1}\frac{\partial o(\bar{\bm\phi}^\star)}{\partial\bar{\bm\phi}^\star[l]}\\
    &\overset{\mathrm{(a)}}{=}- \left(\frac{\partial o(\bar{\bm\phi}^\star)}{\partial\bar{\bm\phi}^\star[r]}\right)^{\mathrm{T}}\bm\Sigma_{\bar{\bm\eta}}^{-1}\frac{\partial o(\bar{\bm\phi}^\star)}{\partial\bar{\bm\phi}^\star[l]},
\end{aligned}
\end{equation}
and 
\begin{equation}
\begin{aligned}
    &\bm B_{\bar{\bm \phi}^\star}[r,l]\\
    &=\mathbb{E}_{g_{\text{T}}(\hat{\bm\eta}_{\text{Eve}}|\bm\phi^\star)}\left\{\frac{\partial\log g_{\text{M}}(\hat{\bm\eta}_{\text{Eve}}|\bar{\bm\phi}^\star) }{\partial\bar{\bm\phi}^\star[r]}\frac{\partial\log g_{\text{M}}(\hat{\bm\eta}_{\text{Eve}}|\bar{\bm\phi}^\star) }{\partial\bar{\bm\phi}^\star[l]}\right\},\\
    &= \left(\frac{\partial o(\bar{\bm\phi}^\star)}{\partial\bar{\bm\phi}^\star[r]}\right)^{\mathrm{T}}\bm\Sigma_{\bar{\bm\eta}}^{-1}\mathbb{E}_{g_{\text{T}}(\hat{\bm\eta}_{\text{Eve}}|\bm\phi^\star)}\left\{\left(\hat{\bm\eta}_{\text{Eve}}-o(\bar{\bm\phi}^\star)\right)\right.\\
    &\quad\quad\quad\quad\quad\quad\quad\quad\quad\quad\quad\left.\cdot\left(\hat{\bm\eta}_{\text{Eve}}-o(\bar{\bm\phi}^\star)\right)^{\mathrm{T}}\right\}\bm\Sigma_{\bar{\bm\eta}}^{-1}\frac{\partial o(\bar{\bm\phi}^\star)}{\partial\bar{\bm\phi}^\star[l]}\\
    &\overset{\mathrm{(b)}}{=} \left(\frac{\partial o(\bar{\bm\phi}^\star)}{\partial\bar{\bm\phi}^\star[r]}\right)^{\mathrm{T}}\bm\Sigma_{\bar{\bm\eta}}^{-1}\frac{\partial o(\bar{\bm\phi}^\star)}{\partial\bar{\bm\phi}^\star[l]},
\end{aligned} 
\end{equation}
where $\mathrm{(a)}$ and $\mathrm{(b)}$ follow from $o(\bar{\bm\phi}^\star) = u(\bm\phi^{\star}) = \bar{\bm \eta}$ given $\bm\epsilon\sim\mathcal{N}(\bm0,\bm\Sigma_{\bar{\bm\eta}})$. Then, the lower bound for the MSE of Eve's localization can be simplified into 
\begin{equation}\label{eq:EveLBsimplified}
\begin{aligned}
    &\bm\Psi_{\bar{\bm \phi}^\star}={{\bm A_{\bar{\bm \phi}^\star}^{-1}\bm B_{\bar{\bm \phi}^\star}\bm A_{\bar{\bm \phi}^\star}^{-1}}+{(\bar{\bm\phi}^\star-\bm\phi^{\star})(\bar{\bm\phi}^\star-\bm\phi^{\star})^\mathrm{T}}}\\
    &={{\left( \left(\frac{\partial o(\bar{\bm\phi}^\star)}{\partial\bar{\bm\phi}^\star}\right)^{\mathrm{T}}\bm\Sigma_{\bar{\bm\eta}}^{-1}\frac{\partial o(\bar{\bm\phi}^\star)}{\partial\bar{\bm\phi}^\star}\right)^{-1}}+{(\bar{\bm\phi}^\star-\bm\phi^{\star})(\bar{\bm\phi}^\star-\bm\phi^{\star})^\mathrm{T}}}\\
    &={{\left( \left(\frac{\partial \bar{\bm \eta}}{\partial\bar{\bm\phi}^\star}\right)^{\mathrm{T}}\bm\Sigma_{\bar{\bm\eta}}^{-1}\frac{\partial \bar{\bm \eta}}{\partial\bar{\bm\phi}^\star}\right)^{-1}}+{(\bar{\bm\phi}^\star-\bm\phi^{\star})(\bar{\bm\phi}^\star-\bm\phi^{\star})^\mathrm{T}}}\\
    &={{\left(\bm\Pi_{\bar{\bm \phi}^{\star}}^\mathrm{T}\bm J_{\bar{\bm \eta}}\bm\Pi_{\bar{\bm \phi}^{\star}}\right)^{-1}}+{(\bar{\bm\phi}^\star-\bm\phi^{\star})(\bar{\bm\phi}^\star-\bm\phi^{\star})^\mathrm{T}}}\\
    &={{\bm J_{\bar{\bm \phi}^{\star}}^{-1}}+{(\bar{\bm\phi}^\star-\bm\phi^{\star})(\bar{\bm\phi}^\star-\bm\phi^{\star})^\mathrm{T}}},
\end{aligned}
\end{equation}
where $\bm\Pi_{\bar{\bm \phi}^{\star}}\triangleq\frac{\partial \bar{\bm \eta}}{\partial\bar{\bm\phi}^\star}$, concluding the proof.
\end{proof}
}

\begin{remark}[Interpretation of the simplified MCRB]
    {The lower bound consists of two parts, {\it i.e.,} ${\bm\Psi_{\bar{\bm \phi}^\star}^{(\romannum{1})}}\triangleq{\bm J_{\bar{\bm \phi}^{\star}}^{-1}}$ and ${\bm\Psi_{\bar{\bm \phi}^\star}^{(\romannum{2})}}\triangleq {(\bar{\bm\phi}^\star-\bm\phi^{\star})(\bar{\bm\phi}^\star-\bm\phi^{\star})^\mathrm{T}}$. Note that $\bm J_{\bar{\bm \phi}^{\star}}$ is the FIM for the estimation of ${\bar{\bm \phi}^{\star}}$. This indicates that using the proposed DAIS scheme we virtually move Alice and all the scatterers to certain incorrect locations (pseudo-true locations) ${\bar{\bm \phi}^\star}$ and Eve is misled to estimate such \textit{pseudo-true} locations; for the inference of Alice's true location, Eve has to address the estimation error of ${\bar{\bm \phi}^\star}$, {\it i.e.}, ${\bm\Psi_{\bar{\bm \phi}^\star}^{(\romannum{1})}}$, as well as the associated geometric mismatch, {\it i.e.}, ${\bm\Psi_{\bar{\bm \phi}^\star}^{(\romannum{2})}}$. }
\end{remark}

\subsection{Design of $\Delta_\tau$ and $\Delta_\theta$}
{
In this subsection, we will provide a closed-form expression for the simplified MCRB in Corollary \ref{coro:EveLBsimplifed} under a mild condition, suggesting the appropriate values of $\Delta_\tau$ and $\Delta_\theta$ when the CSI is available. 
\begin{proposition}\label{prop:lowerboundpos}
Denote by $\hat{\bm p}_\text{Eve}$ the position estimated by Eve. If the paths of mmWave MISO OFDM channels are orthogonal to each other in the following sense\footnote{{Paths are approximately orthogonal to each other for a large number of symbols and transmit antennas due to the low-scattering sparse nature of the mm-Wave channels \cite{li2023fpi}. We will numerically validate the analysis for the design of $\bm\Delta$ in Section \ref{sec:numericalresults} with a practical number of symbols and transmit antennas though orthogonality among paths does not hold strictly.}}, \vspace{-3pt}
\begin{equation}\label{eq:orthopath}
    \frac{2}{\sigma^2}\sum_{n=0}^{N-1}\sum_{g=1}^{G}\mathfrak{R}\left\{\left(\frac{\partial\bar{ u}^{(g,n)}}{\partial \bar{\xi}_k}\right)^{\rev{\mathrm{H}}}\frac{\partial \bar{ u}^{(g,n)}}{\partial \bar{\xi}_{k^{\prime}}}\right\}=0, \quad k\neq k^{\prime}, \vspace{-3pt}
\end{equation}
where $\bar{\xi}_k\in\{\bar{\tau}_k,\bar{\theta}_{\mathrm{Tx},k},\mathfrak{R}\{\gamma^{\star}_k\},\mathfrak{I}\{\gamma^{\star}_k\}\rev{\}}$ and $\bar{\xi}_{k^{\prime}}$ is defined similarly, with $k,k^{\prime}=0,1,\cdots,K$, in the presence of the DAIS proposed in \cite{DAIS}, {for any real $\psi>0$, there always exists a positive integer $\mathcal{G}$ such that when $G\geq\mathcal{G}$} the MSE for Eve's localization can be bounded as, 
\begin{equation}\label{eq:EveLBsimplify}
\begin{aligned}
    \mathbb{E}\left\{\left\|\hat{\bm p}_\text{Eve}-{\bm p}^{\star}\right\|^2_2\right\}&\geq C_1+\frac{C_2\bar{\tau}_{k_{\min}}^2}{\cos^2(\bar{\theta}_{\mathrm{Tx},{k_{\min}}})}+\left\|\bar{\bm p}^\star-{\bm p}^{\star}\right\|^2_2,
\end{aligned}
\end{equation}
with probability of $1$, where 
\begin{subequations}
    \begin{align}
        C_1&\triangleq\frac{3\sigma^2c^2NT_s^2}{2G|\rev{\gamma}_{k_{\min}}|^2\pi^2(N^2-1)}-\frac{\psi}{G},\\
        C_2&\triangleq\frac{3\sigma^2c^2\lambda_c^2}{2G|\rev{\gamma}_{k_{\min}}|^2\pi^2d^2N(N_t^2-1)}.
    \end{align}
\end{subequations}
\end{proposition}
\begin{proof}
With slight abuse of notation, we shall define the vector of the location-relevant channel parameters as $\bar{\bm\eta}\triangleq[\bar{\bm\eta}^{\mathrm{T}}_{k_{\min}
},\bar{\bm\eta}_1^{\mathrm{T}},\cdots,\bar{\bm\eta}_0^{\mathrm{T}},\cdots,\bar{\bm\eta}_K^{\mathrm{T}}]^{\mathrm{T}}\in\mathbb{R}^{2(K+1)\times1}$ with $\bar{\bm\eta}_k\triangleq[\bar{\tau}_k,\bar{\theta}_{\mathrm{Tx},k}]^{\mathrm{T}}\in\mathbb{R}^{2\times1}$ for $k=0,1,\cdots,K$. Given the orthogonality of the paths in Equation \eqref{eq:orthopath}, it can be verified that the effective FIM for the estimation of $\bar{\bm\eta}$ is a block diagonal matrix, 
\begin{equation}\label{eq:blockJ}
    \bm J_{\bar{\bm \eta}}= \begin{bmatrix}\bm J_{\bar{\bm \eta}_{k_{\min}}} & \bm 0 & \cdots & \bm 0\\
    \bm 0 & \bm J_{\bar{\bm \eta}_1} & \cdots & \bm 0\\
    \vdots & \vdots & \ddots & \vdots\\
    \bm 0 & \bm 0 & \cdots & \bm J_{\bar{\bm \eta}_K}\end{bmatrix}. 
\end{equation}
In addition, we can express the inverse of the transformation matrix $\bm\Pi_{\bar{\bm \phi}^{\star}}$ as 
\begin{equation}\label{eq:inversePi}
\begin{aligned}
    &\bm\Pi_{\bar{\bm \phi}^{\star}}^{-1}\overset{\mathrm{(c)}}{=}\left(\frac{\partial\bar{\bm\phi}^\star}{\partial \bar{\bm\eta}}\right)^{\mathrm{T}}=\begin{bmatrix}
    \bm T_{0,0} & \bm T_{1,0} & \cdots & \bm T_{K,0}\\
    \bm 0       & \bm T_{1,1} & \cdots & \bm 0\\
    \vdots      & \vdots      & \ddots & \vdots\\
    \bm 0       & \bm 0       & \cdots & \bm T_{K,K}
    \end{bmatrix},
\end{aligned} 
\end{equation}
where $\mathrm{(c)}$ holds due to the multivariate inverse function theorem \cite{FascistaMISOML} and the matrix $\bm T_{k,k^\prime}$ is given by 
\begin{equation}\label{eq:subT}
    \bm T_{k,k^\prime} \triangleq \left(\frac{\partial\bar{\bm v}_k^\star}{\partial \bar{\bm\eta}_{k^\prime}}\right)^{\mathrm{T}} = \begin{bmatrix} \frac{\partial \bar{v}_{x,k}^\star}{\partial \bar{\tau}_{k^\prime}} & \frac{\partial \bar{v}_{y,k}^\star}{\partial \bar{\tau}_{k^\prime}} \\ \frac{\partial \bar{v}_{x,k}^\star}{\partial \bar{\theta}_{\mathrm{Tx},k^\prime}} & \frac{\partial \bar{v}_{y,k}^\star}{\partial \bar{\theta}_{\mathrm{Tx},k^\prime}}\end{bmatrix},
\end{equation}
for $k,k^{\prime}=0,1,\cdots,K$, with $\bar{\bm v}^\star_0\triangleq\bar{\bm p}^\star$. Then, ${\bm J_{\bar{\bm \phi}^{\star}}^{-1}}$ can be derived as follows, 
\begin{equation}
\begin{aligned}
    {\bm J_{\bar{\bm \phi}^{\star}}^{-1}}= {\left(\bm\Pi_{\bar{\bm \phi}^{\star}}^\mathrm{T}\bm J_{\bar{\bm \eta}}\bm\Pi_{\bar{\bm \phi}^{\star}}\right)^{-1}}= \left(\bm\Pi^{-1}_{\bar{\bm \phi}^{\star}}\right)^\mathrm{T}\bm J^{-1}_{\bar{\bm \eta}}\bm\Pi_{\bar{\bm \phi}^{\star}}^{-1}.
\end{aligned} 
\end{equation}
Following from Corollary \ref{coro:EveLBsimplifed} as well as Equations \eqref{eq:blockJ} and \eqref{eq:inversePi}, it can be verified that  \vspace{-3pt}
\begin{equation}\label{eq:EveLBtbsimplify}
\begin{aligned}
    \mathbb{E}\left\{\left\|\hat{\bm p}_\text{Eve}-{\bm p}^{\star}\right\|^2_2\right\} \geq \rev{\operatorname{Tr}}\left(\bm T^{\mathrm{T}}_{0,0}\bm J^{-1}_{\bar{\bm \eta}_{k_{\min}}}\bm T_{0,0}\right)+\left\|\bar{\bm p}^\star-{\bm p}^{\star}\right\|^2_2.
\end{aligned} \vspace{-3pt}
\end{equation}
According to the definition for the effective FIM given in Equations \eqref{eq:efim} and that for the transformation matrix given in \eqref{eq:subT}, substituting Equation \eqref{eq:ptrue} into Equation \eqref{eq:EveLBtbsimplify} and leveraging the asymptotic property of the FIM shown in \cite[Lemma 1]{li2023fpi} yield the desired statement.
\end{proof}
}

\begin{remark}[Design of $\Delta_\theta$]
    {Given the knowledge of $k_{\min}$ and ${\theta}^{\star}_{\mathrm{Tx},{k_{\min}}}$, the desired choices of $\Delta^{d}_\theta$ is selected as 
    \begin{equation}
    \begin{aligned}
        &\Delta^{d}_\theta \in \mathcal{A}\triangleq\\
        &\left\{\theta|\theta=\arcsin\left(\left(1-\sin(\theta^{\star}_{\mathrm{Tx}, {k_{\min}}})\right)_{\left(-1,1\right]}\right)+2m\pi,m\in\mathbb{Z}\right\} 
    \end{aligned} 
    \end{equation}
    such that $\cos^2(\bar{\theta}_{\mathrm{Tx},{k_{\min}}})$=0 and the lower bound derived in Equation \eqref{eq:EveLBsimplify} goes to infinity, rendering the estimation problem unstable. Under the above design of $\Delta_\theta$, the angle information is not accessible to Eve and thus she cannot accurately infer Alice's position.}
\end{remark}

\begin{remark}[Design of $\Delta_\tau$]
    {According to Proposition \ref{prop:lowerboundpos}, the lower bound on the MSE for Eve's localization is positively related with $\bar{\tau}_{k_{\min}}^2$ and $\left\|\bar{\bm p}^\star-{\bm p}^{\star}\right\|^2_2$. 
    Hence, given the knowledge of $\bm\eta^\star$ and $\Delta_\theta\notin\mathcal{A}$, we can solve the following problem for the optimal design of $\Delta_\tau$,
    \begin{equation}
        \Delta_\tau^{d} = \arg\max_{\Delta_\tau} \frac{C_2\bar{\tau}_{k_{\min}}^2}{\cos^2(\bar{\theta}_{\mathrm{Tx},{k_{\min}}})}+\left\|\bar{\bm p}^\star-{\bm p}^{\star}\right\|^2_2. 
    \end{equation}
    We note that $\left\|\bar{\bm p}^\star-{\bm p}^{\star}\right\|^2_2$ is the distance between Alice's true location and the {spoofed} location. Therefore, in terms of the location-privacy enhancement, if we maximally increase $\left\|\bar{\bm p}^\star-{\bm p}^{\star}\right\|^2_2$, then we virtually move Alice as far away from her true location as possible.}
\end{remark}

\section{Numerical Results}\label{sec:numericalresults}
{In this section,  we validate the theoretical analysis through the evaluation of the lower bound derived in \ref{sec:simplifiedMCRB}, which bounds the smallest possible MSE for Eve employing a mis-specified-unbiased estimator. 
The evaluated RMSE for Eve’s localization is defined as $\operatorname{RMSE}_{\text{Eve}}\triangleq\sqrt{\bm\Psi_{\bar{\bm \phi}^\star}[1,1]+\bm\Psi_{\bar{\bm \phi}^\star}[2,2]}. $
In the numerical results, we place Alice and Eve at $[3 \text{ m},0 \text{ m}]^{\mathrm{T}}$ and $[10 \text{ m},5 \text{ m}]^{\mathrm{T}}$, respectively. while the pilot signals are randomly and uniformly generated on the unit circle. We set the parameters $K$, $\varphi_c$, $B$, $c$, $N$, $G$, and $N_t$, to $2$, $60$ GHz, $30$ MHz, $300$ m/$\mu s$, $16$, $16$, and $16$, respectively. A three-path channel is considered, with two scatterers at $[8.87\text{ m}, -6.05 \text{ m}]^{\mathrm{T}}$ ($k=1$) and $[7.44 \text{ m}, 8.53 \text{ m}]^{\mathrm{T}}$ ($k=2$), respectively. SNR is set to $20$ dB.}

\begin{figure}[t]
\centering
\includegraphics[scale=0.49]{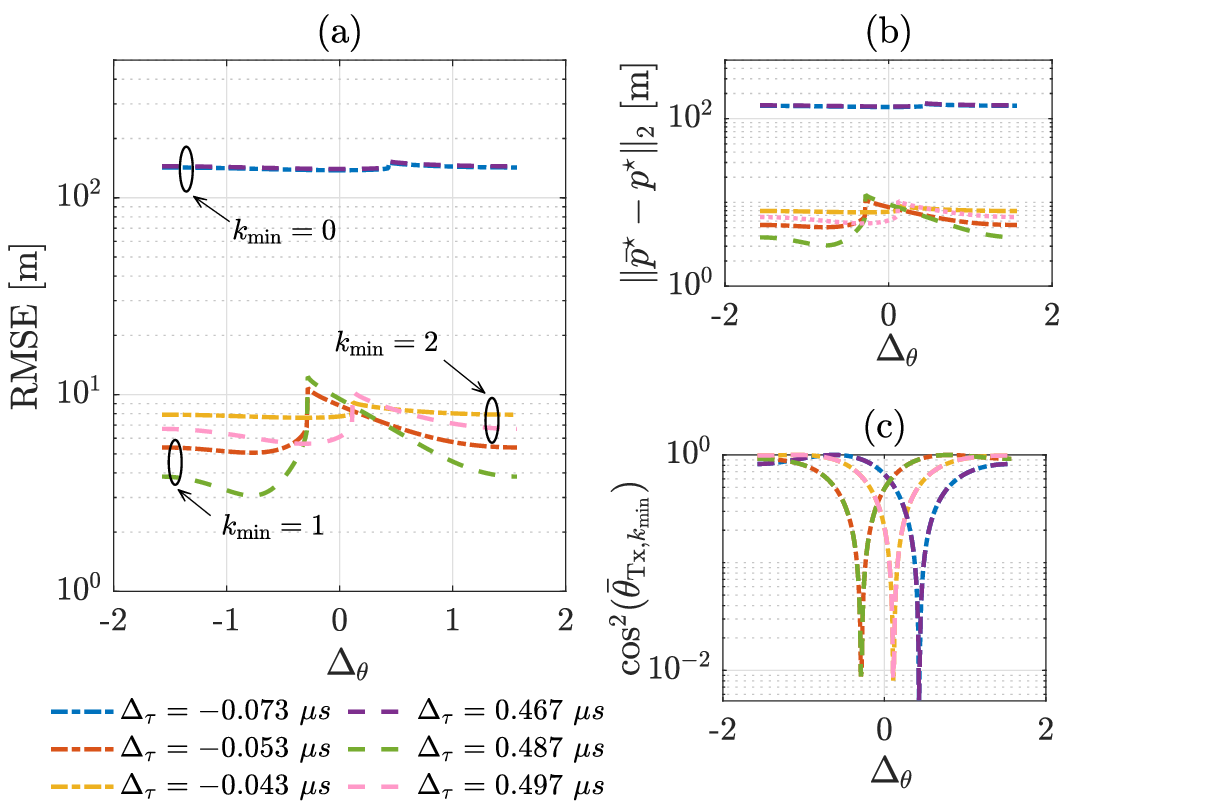}\vspace{-5pt}
\caption{{(a) Lower bounds for the RMSE of Eve’s localization with different choices of $\bm\Delta$, where $k_{\min}=k$ means that the $k$-th path with the smallest shifted TOA is assumed to be the LOS path for localization due to phase wrapping according to Equation \eqref{eq:ptrue}; (b) The associated values of $\|\bar{\bm p}^{\star}-\bar{\bm p}\|$; (c) The associated values of $\cos^2(\bar{\theta}_{\mathrm{Tx},k_{\min}})$.}}
\label{fig:rmse_angle} \vspace{-15pt}
\end{figure}

{Under different choices of $\Delta_\tau$ and $\Delta_\theta$,
the associated RMSEs for Eve's localization accuracy are shown in Figure \ref{fig:rmse_angle} (a). As observed in Figure \ref{fig:rmse_angle} (a), variations of the RMSE are more pronounced when we change the value of $\Delta_\tau$, which suggests the adjustment of $\Delta_\theta$ has a relatively less influence on Eve’s performance degradation if $\cos^2(\bar{\theta}_{\mathrm{Tx},k_{\min}})$ is not close to zero. {We wish to underscore the non-linear and non-monotonic effect of changing $\Delta_\tau$ on the RMSE; hence the need for the analysis to determine the proper values of the spoofing parameters. Note that the RMSE can be driven up to $150$ m with proper choice of spoofing parameters. We contrast this with our numerical results in \cite{DAIS} where our ``less informed'' choices achieved an error of about $90$ m.}

In addition, for the given choices of $\Delta_\tau$ and $\Delta_\theta$, we show the corresponding values for $\|\bar{\bm p}^{\star}-\bar{\bm p}\|$ and $\cos^2(\bar{\theta}_{\mathrm{Tx},k_{\min}})$ in Figure \ref{fig:rmse_angle} (b) and (c), respectively. We note that the orthogonality among the paths does not strictly hold in the numerical results. However, coinciding with our theoretical analysis of the parameter designs, relatively higher RMSEs are still achieved at large values of $\|\bar{\bm p}^{\star}-\bar{\bm p}\|$ and small values of $\cos^2(\bar{\theta}_{\mathrm{Tx},k_{\min}})$. Therefore, to significantly degrade Eve's localization accuracy using the DAIS strategy, one option is to increase the distance between Alice's true and pseudo-true locations as much as possible with the adjustment of the design parameters $\Delta_\tau$ and $\Delta_\theta$; under the knowledge of $k_{\min}$ and ${\theta}^{\star}_{\mathrm{Tx},{k_{\min}}}$, another option is to design $\Delta_\theta$ to make the line between Alice's pseudo-true location and Eve's location parallel to Alice's antenna array.
}

\vspace{-5pt}
\section{Conclusions}\vspace{-2.5pt}
{An augmented analysis was provided for location-privacy enhancement with a delay-angle information spoofing design. A simplified lower bound on eavesdropper's localization error was derived, theoretically validating that the eavesdropper can be misled to mistake the incorrect location as the true one. Furthermore, the simplified lower bound was explicitly expressed as a function of the design parameters under a mild condition on the orthogonality between the paths in a multi-path channel. This expression provided appropriate choices of these key parameters to maximize spoofing. With an adjustment of the design parameters according to the analysis, the root-mean-square error for eavesdropper's localization can be up to 150 m,  providing very effective location privacy.}


\appendices
\begingroup
\allowdisplaybreaks

\vspace{-5pt}
\renewcommand*{\bibfont}{\footnotesize}
\printbibliography
\end{document}